\documentclass[12pt]{l4dc2020}

\usepackage{amsmath}

\DeclareMathOperator*{\argmin}{arg\,min}

\usepackage{enumitem}

\title[lambda PI]{Lambda-Policy Iteration with Randomization for Contractive Models\\with Infinite Policies: Well-Posedness and Convergence\\(Extended Version)}
\usepackage{times}

\author{%
 \Name{Yuchao Li} \Email{yuchao@kth.se}\\
 \Name{Karl H. Johansson} \Email{kallej@kth.se}\\
 \Name{Jonas M\aa rtensson} \Email{jonas1@kth.se}\\
 \addr Division of Decision and Control Systems, KTH Royal Institute of Technology, Stockholm, Sweden%
}

\begin{document}

\maketitle

\begin{abstract}%
Abstract dynamic programming models are used to analyze
$\lambda$-policy iteration with randomization algorithms. Particularly, contractive models with infinite policies are considered and it is shown that well-posedness of the $\lambda$-operator plays a central role in the algorithm. The operator is known to be well-posed for problems with finite states, but our analysis shows that it is also well-defined for the contractive models with infinite states studied. Similarly, the algorithm we analyze is known to converge for problems with finite policies, but we identify the conditions required to guarantee convergence with probability one when the policy space is infinite regardless of the number of states. Guided by the analysis, we exemplify a data-driven approximated implementation of the algorithm for estimation of optimal costs of constrained linear and nonlinear control problems. Numerical results indicate potentials of this method in practice.         
\end{abstract}


\section{Introduction}
Temporal-difference (TD) learning is a prominent class of algorithms widely applied in reinforcement learning (RL). Its first formal treatment is given in \cite{sutton1988learning} where a family of algorithms, collectively known as TD($\lambda$), is analyzed in the context of absorbing Markov processes. By utilizing the properties of transitional matrices of the process, algorithm convergence guarantees are established. Structural relations between RL and dynamic programming (DP) was noted by \cite{watkins1989learning}, and foundations for the understanding of RL followed. The monograph by \cite{bertsekas1996neuro} puts a broad class of RL algorithms in the context of two principle methods of DP, viz., value iteration (VI) and policy iteration (PI), and collects a bundle of research outputs of interests.\footnote{A detailed document of the history can be found in \cite[Chapter 1]{sutton2018reinforcement}.} Among those results, the analysis of TD($\lambda$), originally given in \cite{bertsekas1996temporal}, unveils the underlying DP problem of TD($\lambda$). As is shown, the desired behavior of TD($\lambda$) is inherited from the parameter $\lambda$ being a discount factor in the classical DP sense and the infinite iterates of TD algorithms can be interpreted as an iteration of a compactly defined operator. In addition, TD($\lambda$) can be embedded into the PI framework, which is then named $\lambda$-PI. There has been a tremendous development in algorithms related to $\lambda$-PI, such as \cite{scherrer2010least,scherrer2015approximate}. A survey can be found in \cite{bertsekas2012lambda}. Most recently, the connection between TD($\lambda$) and proximal algorithms, which are widely used for solving convex optimization problems, is discussed in \cite{bertsekas2018proximal}. In light of such relation, $\lambda$-PI with randomization ($\lambda$-PIR) was proposed in \cite[Chaper 2]{bertsekas2018abstract}. The algorithm resembles the one proposed in \cite{yu2015mixed}, and offers a scheme to combine the fast computations by proximal algorithms with the convergence behavior by VI. Apart from these algorithmic properties, the abstract approach taken for analyzing $\lambda$-PIR is also well worth special attention. Although some operators, in particular the Bellman operator, are often used in algorithmic analysis, they played less of a central role throughout the development, cf. \cite{tsitsiklis1997analysis,de2003linear,wang2015approximate,bellemare2016increasing,bian2016value,banjac2019data}, in which operator computations are utilized while specific properties of the problem are also taken advantage of. An exception is the analysis of $\lambda$-PIR in \cite[Chaper 2]{bertsekas2018abstract}, which has solely relied on abstract operator properties. There are many advantages of such an approach, e.g., (a) it can single out the key factor that stands behind the desired behavior of the algorithm; (b) it can shed new lights on the understanding of some algorithms and help to bring together isolated methods; (c) it can help to safeguard the desired behaviors when modifying and generalizing algorithms. One example of this is by \cite{yu2018generalized}, in which the parameter $\lambda$ is extended to be state-dependent, while fundamental properties are still guaranteed.

In this paper, we use abstract DP models and extend $\lambda$-PIR for finite policy problems \cite[Chapter 2]{bertsekas2018abstract} to contractive models with infinite policies. A policy space can be infinite due to infinite states, or infinite control over some finite state space. We make the following main contributions:
\begin{itemize}[noitemsep,topsep=0pt]
    \item[(1)] We establish the well-posedness of the compact operator that plays a central role in the algorithm (Theorem~\ref{thm:2}). Our result relies solely on the contraction property of the model.
    \item[(2)] Conditions for convergence of $\lambda$-PIR for problems with infinite policies are given (Theorem~\ref{thm:42}). We show that such conditions can be dismissed if the underlying operator exhibits a linear structure (Corollary~\ref{cor:1}).
\end{itemize}

The rest of the paper is organized as follows: Section~\ref{sec:preli} gives a brief account of preliminaries of contractive models and existing results on $\lambda$-PIR. Section~\ref{sec:acma} presents results on well-posedness of the $\lambda$-PIR algorithm for infinite-state problems. Conditions for convergence of $\lambda$-PIR for problems with infinite policies are given in Section~\ref{sec:conv}. Section~\ref{sec:aadp} explains an approximated implementation of $\lambda$-PIR and shows its application when embedded in the approximate dynamic programming (ADP) framework. Section~\ref{sec:con} concludes the paper.

\section{Preliminaries}\label{sec:preli}

Here we introduce the concepts and some preliminaries related to contractive models, and the $\lambda$-PIR algorithm. The contents here are mostly taken from \cite{bertsekas2018abstract}. 

\subsection{Contractive models}
Given a state space $X$, a control space $U$, and for each $x\in X$ a nonempty control set $U(x)\subset U$, we denote $\mathcal{M}=\{\mu\,|\,\mu(x)\in U(x),\,\forall x\in X\}$ and name it as the set of policies whose elements are denoted by $\mu$. One can see that the set $\mathcal{M}$ can be viewed as the Cartesian product $\prod_{x\in X}U(x)$. We denote by $\mathcal{R}(X)$ the set of functions $J:X\to \mathbb{R}$ and by $\mathcal{E}(X)$ the set of functions $J:X\to \mathbb{R}^*$ where $\mathbb{R}^*=\mathbb{R}\cup \{\infty,-\infty\}$. We study the mappings of the form $H:X\times U\times \mathcal{R}(X)\to \mathbb{R}$. For every $\mu\in \mathcal{M}$, we define $T_\mu:\mathcal{R}(X)\to\mathcal{R}(X)$ as
\begin{equation}
    (T_\mu J)(x)=H(x,\mu(x),J),\;\forall x\in X,
\end{equation}
and the mapping $T:\mathcal{R}(X)\to\mathcal{E}(X)$ as
\begin{equation}
\label{eq:tdef}
    (T J)(x)=\inf_{\mu\in \mathcal{M}} (T_\mu J)(x),\;\forall x\in X.
\end{equation}
In view of the definitions of $\mathcal{M}$, $T_\mu$, and $T$, we have
\begin{equation}
\label{eq:cart}
    (T J)(x)=\inf_{u\in U(x)} H(x,u,J)=\inf_{\mu\in \mathcal{M}} H(x,\mu(x),J).
\end{equation}
Given some positive function $v:X\to\mathbb{R}$, we denote by $\mathcal{B}(X)$ the set of functions $J$ such that $\sup_{x\in X}\frac{|J(x)|}{v(x)}<\infty$. We define a norm $\Vert\cdot\Vert$ on $\mathcal{B}(X)$ as
\begin{equation*}
    \Vert J\Vert =\sup_{x\in X}\frac{|J(x)|}{v(x)}.
\end{equation*}

The following lemmas are classical results from functional analysis. The proof of the first can be found in \cite[Appendix B]{bertsekas2018abstract}, while the second is explained in \cite[Appendix A]{szepesvari2010algorithms}.
\begin{lemma}
$\mathcal{B}(X)$ is complete with respect to the metric induced by $\Vert\cdot\Vert$.
\end{lemma}

\begin{lemma}\label{thm:pre2}
Given a sequence $\{J_k\}\subset \mathcal{B}(X)$ and $J\in \mathcal{B}(X)$, if $J_k\to J$ in the sense that $\lim_{k\to\infty}\Vert J_k-J\Vert =0$, then $\lim_{k\to\infty}J_k(x)= J(x),\,\forall x\in X$. 
\end{lemma}
\begin{remark}
The converse of Lemma~\ref{thm:pre2} does not necessarily hold, see \cite[Appendix A]{szepesvari2010algorithms}.
\end{remark}

For the mappings $H$, $T_\mu$ and $T$ on $\mathcal{B}(X)$, we introduce the following standard assumptions. 
\begin{assum}[Well-posedness]\label{asm:pose}
$\forall J\in\mathcal{B}(X)$ and $\forall \mu\in\mathcal{M}$, $T_\mu J\in\mathcal{B}(X)$ and $T J\in\mathcal{B}(X)$.
\end{assum}

\begin{assum}[Uniform contraction]\label{asm:contra}
For some $\alpha\in(0,1)$, it holds that 
\begin{equation*}
    \Vert T_\mu J-T_\mu J'\Vert\leq \alpha \Vert J-J'\Vert,\;\forall J,J'\in\mathcal{B}(X),\,\mu\in\mathcal{M}.
\end{equation*}
\end{assum}
One immediate consequence of Assumption~\ref{asm:contra} is that $T$ is also a contraction with the same modulus $\alpha$, see \cite[Chapter 1]{bertsekas2018abstract}. When Assumptions~\ref{asm:pose} and \ref{asm:contra} hold, the following convergence result holds due to the fixed point theory.
\begin{lemma}[\cite{bertsekas2018abstract}, Proposition B.1]\label{thm:con}
Let Assumptions~\ref{asm:pose}, and \ref{asm:contra} hold. Then:
\begin{itemize}
    \item[(a)] There exist unique $J_\mu,J^*\in \mathcal{B}(X)$ such that $TJ^*=J^*;\;T_\mu J_\mu=J_\mu,\,\forall \mu\in \mathcal{M}$.
    \item[(b)] For arbitrary $J_0\in \mathcal{B}(X)$, the sequence $\{J_k\}$ where $J_{k+1}=T_\mu J_k$ converges in norm to $J_\mu$. 
    \item[(c)] For arbitrary $J_0\in \mathcal{B}(X)$, the sequence $\{J_k\}$ where $J_{k+1}=T J_k$ converges in norm to $J^*$. 
\end{itemize}
\end{lemma}

The above results are the backbones of VI. However, they do not guarantee the effectiveness of PI, for which we need some additional assumptions.  

\begin{assum}[Monotonicity]\label{asm:mon}
$\forall J,J'\in\mathcal{B}(X)$, it holds that $J\leq J'$ implies $H(x,u,J)\leq H(x,u,J'),\;\forall x\in X,\,u\in U(x)$, where $\leq$ indicates point-wise relation.
\end{assum}

\begin{assum}[Attainability]\label{asm:att}
For all $ J\in\mathcal{B}(X)$, there exists $\mu\in\mathcal{M}$, such that $T_\mu J=TJ$.
\end{assum}

In fact, only after including Assumption~\ref{asm:mon}, in addition to Assumptions~\ref{asm:pose} and \ref{asm:contra}, can $J^*$ be interpreted as optimal in the sense that $J^*(x)=\inf_{\mu\in\mathcal{M}}J_\mu(x)$. Besides, due to the nature of $\mathcal{M}$ being a Cartesian product of feasible control sets $U(x)$, for arbitrary small $\varepsilon>0$, we can always construct an $\varepsilon$-optimal policy $\mu_\varepsilon\in\mathcal{M}$ in the sense that $J_{\mu_\varepsilon}(x)\leq J^*(x)+\varepsilon$ holds for all $x$. One such construction in a more general setting can be found in \cite[Chapter 2]{bertsekas1978stochastic} and the details of the above discussion can be found in \cite[Propositions 2.1.1, 2.1.2]{bertsekas2018abstract}. Since the infimum in \eqref{eq:tdef} is not always attained, Assumption~\ref{asm:att} is needed for PI-based methods. In the subsequent sections, we always assume Assumption~\ref{asm:pose} hold, and therefore do not repeat it in all the theoretical statements.

\subsection{$\lambda$-PIR}\label{sec:ag}
The following $\lambda$-PIR algorithm is introduced in \cite[Chapter 2]{bertsekas2018abstract} in the abstract setting. Given some $\lambda\in[0,1)$, consider the mappings $T_\mu^{(\lambda)}$ with domain $\mathcal{B}(X)$ and defined point-wise by 
\begin{equation}
	\label{eq:tlambda}
	\big(T_\mu^{(\lambda)}J\big)(x)=(1-\lambda)\sum_{\ell=1}^\infty  \lambda^{\ell-1}\big(T_\mu^{\ell}J\big)(x),
\end{equation}
where $T^\ell_\mu$ denotes the $\ell$-fold composition of the operator $T_\mu$, and we refer to the operator $T_\mu^{(\lambda)}$ as $\lambda$ operator in our discussion. Regarding this operator, we make the following mild assumption, which holds for a broad class of DP problems.
\begin{assum}[Commutativeness]\label{asm:comm}
For every $\mu\in\mathcal{M}$, its corresponding $\lambda$ operator and $T_\mu$ commute, viz., for all $J\in\mathcal{B}(X)$, it holds that
\begin{equation*}
    T_\mu\big(T_\mu^{(\lambda)}J\big)=T_\mu^{(\lambda)}(T_\mu J).
\end{equation*}
\end{assum} 
Given $J_k\in\mathcal{B}(X)$ and $p_k\in (0,1)$, then the policy $\mu^k$ and cost approximate $J_{k+1}$ is computed as 
\begin{equation}
\label{eq:algm}
    T_{\mu^k}J_k = TJ_k;\,J_{k+1} =\begin{cases}
    T_{\mu^k}J_k,\,&\text{with prob. }p_k,\\
    T_{\mu^k}^{(\lambda)}J_k,\,&\text{with prob. }1-p_k,
    \end{cases}
\end{equation}
where the policy improvement step to the left is the same as in classical PI, while the evaluation step on the right is a randomized mix between VI and TD learning.  

We list the central statements related to $\lambda$-PIR presented in \cite[Chapter 2]{bertsekas2018abstract}, which include the assumptions needed and convergence behavior of the algorithm. Except the cases in which $U(x)$ is not singleton for finite number of $x$, which we refer to as trivial cases, $\mathcal{M}$ being finite implies state space being finite. Therefore, except the trivial cases, with the following finite policy assumption, the $\lambda$ operator $T^{(\lambda)}_\mu$ is ensured to be well-posed (see \cite[Proposition 2.1]{bertsekas2018proximal}), and the monotonicity of the underlying operator $H$ is not required for the desired behavior. 
\begin{assum}[Finiteness]\label{asm:fit}
$\mathcal{M}$ is finite.
\end{assum}
Then, the following result holds.
\begin{theorem}[\cite{bertsekas2018abstract}, Section 2.5.3]\label{thm:pre}
Let Assumptions~\ref{asm:contra}, \ref{asm:att}, and \ref{asm:fit} hold. $\forall J_0\in \mathcal{B}(X)$, the sequence $\{J_k\}$ generated by $\lambda$-PIR \eqref{eq:algm} converges in norm to $J^*$ with probability one.
\end{theorem}

\section{Well-posedness of $T^{(\lambda)}_\mu$}\label{sec:acma}

We first show a general result, and then show that well-posedness of $T_\mu^{(\lambda)}$ is a consequence of it. For the more general operator, we prove first the output of the operator is well-defined within $\mathcal{R}(x)$, viz., point-wise limits do exist in $\mathbb{R}$. Then we show the output function scaled by the weight function $v(x)$ is bounded, which means that it is an element of $\mathcal{B}(X)$. Then we show the $\lambda$ operator $T^{(\lambda)}_\mu$ is a special case of the proved results. In addition, we explore the relation between the operator defined point-wise and the one by functional sequence, and give an illustrative example to show the difference between them.
\begin{lemma}\label{thm:1}
	Let the set of mappings $T_\mu:\mathcal{B}(X)\to \mathcal{B}(X)$, $\mu\in\mathcal{M}$, satisfy Assumption~\ref{asm:contra}. Consider the mappings $T_\mu^{(w)}$ with domain $\mathcal{B}(X)$ defined point-wise by 
	\begin{equation}
	\label{eq:thm1def}
	\big(T_\mu^{(w)}J\big)(x)=\sum_{\ell=1}^\infty  w_\ell(x)\big(T_\mu^{\ell}J\big)(x),\;x\in X,\,J\in\mathcal{B}(X),
	\end{equation}
	where $w_\ell(x)$ are nonnegative scalars such that for all $x\in X$, $\sum_{\ell=1}^\infty w_\ell(x)=1$. Then the mapping $T_\mu^{(w)}$ is well-defined; namely, for all $x\in X$, $J\in\mathcal{B}(X)$, the sequence 
	\begin{equation}
	\label{eq:thm1}
	\Big\{\sum_{\ell=1}^n  w_\ell(x)\big(T_\mu^{\ell}J\big)(x) \Big\}_{n=1}^\infty
	\end{equation}
	converges with a limit in $\mathbb{R}$, viz., $T_\mu^{(w)}:\mathcal{B}(X)\to\mathcal{R}(X)$.
\end{lemma}

\begin{proof}
	Since $T_\mu$ is a contraction, we have $\big(T_\mu^\ell J\big)(x)\to J_\mu(x)\in\mathbb{R},\,\forall x\in X$ due to Lemma~\ref{thm:pre2}. Therefore, $\big\{\big(T_\mu^\ell J\big)(x)\big\}_{\ell=1}^\infty$ is bounded. Denote the bound as $M_\mu(x)\in\mathbb{R}$. Then $\forall n$, it holds that 
	\begin{align*}
	    \big|\sum_{\ell=1}^n  w_\ell(x)\big(T_\mu^{\ell}J\big)(x)\big|&\leq \sum_{\ell=1}^n w_\ell(x)\big|  \big(T_\mu^{\ell}J\big)(x)\big|\\
	    &\leq \sum_{\ell=1}^n w_\ell(x)M_\mu(x)\\
	    &\leq M_\mu(x)
	\end{align*}
	namely the sequence of \eqref{eq:thm1} is bounded.
	If $J_\mu(x)>0$, then $\exists N$ such that $\big(T_\mu^\ell J\big)(x)>0$ $\forall \ell>N$. Therefore, $\big\{\sum_{\ell=1}^n  w_\ell(x)\big(T_\mu^{\ell}J\big)(x) \big\}_{n=N}^\infty$ is monotonically nondecreasing and bounded by $M_\mu(x)$. Therefore the sequence \eqref{eq:thm1} converges with the limit $\sum_{\ell=1}^\infty  w_\ell(x)\big(T_\mu^{\ell}J\big)(x)\in \mathbb{R}$. If $J_\mu(x)<0$, similar arguments applies. If $J_\mu(x)=0$, then $\forall \varepsilon$, $\exists N$ such that $\forall \ell>N$, $\big|\big(T_\mu^{\ell}J\big)(x)\big|<\varepsilon$. Therefore, $\forall k$, it holds that
	\begin{align}
	&\Big|\sum_{\ell=1}^N w_\ell(x)\big(T_\mu^{\ell}J\big)(x)-\sum_{\ell=1}^{N+k} w_\ell(x)\big(T_\mu^{\ell}J\big)(x)\Big|\nonumber\\
	=&\Big|\sum_{\ell=N+1}^{N+k} w_\ell(x)\big(T_\mu^{\ell}J\big)(x)\Big|\nonumber\\
	\leq& \sum_{\ell=N+1}^{N+k} w_\ell(x)\big|\big(T_\mu^{\ell}J\big)(x)\big|\nonumber\\
	\leq& \sum_{\ell=N+1}^{N+k} w_\ell(x)\varepsilon\nonumber\\
	\leq& \varepsilon,\nonumber
	\end{align}
	which implies that the sequence \eqref{eq:thm1} is Cauchy. As a result, sequence \eqref{eq:thm1} converges in $\mathbb{R}$. Therefore, $\forall J\in\mathcal{B}(X)$, $x\in X$, sequence \eqref{eq:thm1} converges in $\mathbb{R}$. Namely $T_\mu^{(w)}:\mathcal{B}(X)\to\mathcal{R}(X)$. 
\end{proof}

\begin{theorem}\label{thm:2}
	Let the set of mappings $T_\mu:\mathcal{B}(X)\to \mathcal{B}(X)$, $\mu\in\mathcal{M}$, satisfy Assumption~\ref{asm:contra}. Consider the mappings $T_\mu^{(w)}:\mathcal{B}(X)\to\mathcal{R}(X)$ defined in Eq.~\eqref{eq:thm1def}. Then the range of $T_\mu^{(w)}$ is a subset of $\mathcal{B}(X)$, viz., $T_\mu^{(w)}:\mathcal{B}(X)\to\mathcal{B}(X)$; and $T_\mu^{(w)}$ is a contraction.
\end{theorem}

\begin{proof}
	Due to Lemma~\ref{thm:1}, $\forall J\in\mathcal{B}(X)$ and $x\in X$, $\big(T_\mu^{(w)}J\big)(x)$ is well-defined and is a real value. In particular, for $J=J_\mu$, we have $T_\mu^{(w)}J_\mu=J_\mu$ (one may verify this equality by checking the definition Eq.~\eqref{eq:thm1def}). Then we have 
	\begin{align*}
	&\big|\big(T_\mu^{(w)}J\big)(x)-J_\mu(x) \big|\\
	=&\Big|\sum_{\ell=1}^\infty w_\ell(x)\big(T_\mu^{\ell}J\big)(x)-J_\mu(x) \Big|\\
	=&\Big|\sum_{\ell=1}^\infty w_\ell(x)\big(T_\mu^{\ell}J\big)(x)-\sum_{\ell=1}^{\infty} w_\ell(x)\big(T_\mu^{\ell}J_\mu\big)(x) \Big|\\
	=&\Big|\sum_{\ell=1}^\infty w_\ell(x)\Big(\big(T_\mu^{\ell}J\big)(x)- \big(T_\mu^{\ell}J_\mu\big)(x)\Big) \Big|\\
	\leq& \sum_{\ell=1}^\infty w_\ell(x) \big| \big(T_\mu^{\ell}J\big)(x)- \big(T_\mu^{\ell}J_\mu\big)(x)\big|.
	\end{align*}
	Since $T_\mu$ is a contraction, $\forall \ell$, it holds that 
	$$\big| \big(T_\mu^{\ell}J\big)(x)- \big(T_\mu^{\ell}J_\mu\big)(x)\big|\leq \alpha^\ell \Vert J-J_\mu\Vert v(x).$$
	Therefore, we have 
	\begin{align}
	\label{eq:thm2}
	\big|\big(T_\mu^{(w)}J\big)(x)-J_\mu(x) \big|&\leq \sum_{\ell=1}^\infty w_\ell(x) \alpha^\ell \Vert J-J_\mu\Vert v(x)\nonumber\\
	&\leq \Bar{\alpha} \Vert J-J_\mu\Vert v(x)
	\end{align}
	where $\Bar{\alpha}$ is given as
	\begin{equation}
	\label{eq:wcon}
	    \Bar{\alpha}=\sup_{x\in X}\sum_{\ell=1}^\infty w_\ell(x) \alpha^\ell \leq \alpha.
	\end{equation}
	Note that for all $x\in X$, the sequence $\big\{\sum_{\ell=1}^n w_\ell(x) \alpha^\ell\big\}_{n=1}^\infty$ converges in real since it's monotonically nondecreasing and upper bounded by $\alpha$. Therefore $\bar{\alpha}$ is well-defined. Due to triangular inequality, from Eq.~\eqref{eq:thm2}, we have 
	$$\frac{\big|\big(T_\mu^{(w)}J\big)(x) \big|}{v(x)}\leq \Bar{\alpha} \Vert J-J_\mu\Vert +\frac{|J_\mu(x) |}{v(x)}.$$
	Take supremum over $x$ on both sides and due to $J_\mu\in \mathcal{B}(X)$, we have $T_\mu^{(w)}J\in\mathcal{B}(X)$. The contraction proof can be found in \cite[Exercise 1.3]{bertsekas2018abstract} and \cite{yu2012weighted}.
\end{proof}

\begin{corollary}\label{cor:53}
Let the set of mappings $T_\mu:\mathcal{B}(X)\to \mathcal{B}(X)$, $\mu\in\mathcal{M}$, satisfy Assumption~\ref{asm:contra}. The operator $T^{(\lambda)}_\mu$ defined point-wise by Eq.~\eqref{eq:tlambda} is well-posed in the sense that $T^{(\lambda)}_\mu J\in\mathcal{B}(x)$ for all $J\in\mathcal{B}(x)$, and $T^{(\lambda)}_\mu$ is a contraction with modulus $\alpha_\lambda=\alpha(1-\lambda)/(1-\lambda\alpha)$.
\end{corollary}

\begin{proof}
By setting $w_\ell(x)=\lambda^{\ell-1}(1-\lambda)$ for all $x\in X$, it holds that $\sum_{\ell=1}^\infty w_\ell(x)=1$. In view of Theorem~\ref{thm:2}, we have that $T^{(\lambda)}_\mu:\mathcal{B}(X)\to\mathcal{B}(X)$ is a contraction. In addition, by Eq.~\eqref{eq:wcon}, its contraction modulus can be computed as 
\begin{equation}
    \alpha_\lambda=(1-\lambda)\sum_{\ell=1}^\infty \lambda^{\ell-1}\alpha^\ell =\frac{\alpha(1-\lambda)}{1-\lambda\alpha}.
\end{equation}
\end{proof}

The following result shows that the operator $T^{(\lambda)}_\mu$ defined point-wise is no difference compared one defined by convergence in norm.
\begin{lemma}\label{thm:53}
	Let the set of mappings $T_\mu:\mathcal{B}(X)\to \mathcal{B}(X)$, $\mu\in\mathcal{M}$, satisfy Assumption~\ref{asm:contra}. Consider sequence $\{T^{(\lambda_n)}_\mu J\}$ defined by
	\begin{equation*}
	    T^{(\lambda_n)}_\mu J=(1-\lambda)\sum_{\ell=1}^n  \lambda^{\ell-1}T_\mu^{\ell}J.
	\end{equation*}
	The sequence $\{T^{(\lambda_n)}_\mu J\}$ converges to some element $T^{(\lambda_\infty)}_\mu J\in\mathcal{B}(X)$. In addition, it coincides with the function $T^{(\lambda)}_\mu J$ defined by point-wise limit, viz., $T^{(\lambda_\infty)}_\mu J=T^{(\lambda)}_\mu J$.
\end{lemma}
\begin{proof}
Since $\lim_{n\to\infty}\Vert T^n_\mu J-J_\mu\Vert=0$, we have $\lim_{n\to\infty}\Vert T^n_\mu J\Vert=\Vert J_\mu\Vert$. Therefore $\{\Vert T^n_\mu J\Vert\}$ is bounded. Denote its bound as $M_\mu$. Therefore, $\forall \varepsilon$, $\exists N$ such that $\forall k$ 
\begin{align*}
	&\Vert T^{(\lambda_N)}_\mu J- T^{(\lambda_{N+k})}_\mu J\Vert\nonumber\\
	=&\Vert (1-\lambda)\sum_{\ell=1}^N  \lambda^{\ell-1}T_\mu^{\ell}J- (1-\lambda)\sum_{\ell=1}^{N+k}  \lambda^{\ell-1}T_\mu^{\ell}J\Vert\nonumber\\
	=& \Vert (1-\lambda)\sum_{\ell=N+1}^{N+k}  \lambda^{\ell-1}T_\mu^{\ell}J\Vert\nonumber\\
	\leq& (1-\lambda)\sum_{\ell=N+1}^{N+k} \lambda^{\ell-1}\Vert T_\mu^{\ell}J\Vert\nonumber\\
	\leq& (1-\lambda)\sum_{\ell=N+1}^{N+k} \lambda^{\ell-1}M_\mu\nonumber\\
	\leq&\lambda^N M_\mu\nonumber\\
	\leq&\varepsilon,
	\end{align*}
	which implies $\{T^{(\lambda_n)}_\mu J\}$ is Cauchy. Since $\mathcal{B}(X)$ is complete, then it is also convergent. Denote its limit as $T^{(\lambda_\infty)}_\mu J$. Since convergence in norm implies point-wise convergence and limit in $\mathbb{R}$ is unique, then $\forall x\in X$, it holds that $\big(T^{(\lambda_\infty)}_\mu J\big)(x)=\big(T^{(\lambda)}_\mu J\big)(x)$
\end{proof}
Note that the above result does not stand for the more general operator $T^{(w)}_\mu$, when $X$ has infinite cardinality. The following is an example.
\begin{exmp}
Given $X=\{1,\,2,\,...\}$ and define $w_\ell(x)$ as 
\begin{equation*}
    w_\ell(x)=0,\;\ell\leq x,\;\sum_{\ell=x+1}^\infty w_\ell(x)=1.
\end{equation*}
Further we assume that $v(x)=x$. Define $T_\mu:\mathcal{B}(X)\to\mathcal{B}(X)$ as
\begin{equation*}
    (T_\mu J)(x)=(1-\alpha)x+\alpha J(x).
\end{equation*}
Then one can verify that $J_\mu(x)=x$. Then consider sequence $\{T^{(w_n)}_\mu J_\mu\}$ defined point-wise as
\begin{equation*}
    (T^{(w_n)}_\mu J_\mu\big)(x)=\sum_{\ell=1}^n  w_\ell(x)\big(T_\mu^{\ell}J_\mu\big)(x).
\end{equation*}
which can be verified to belong to $\mathcal{B}(X)$. Then $\forall n$, it holds that 
\begin{align*}
    \Vert T^{(w_n)}_\mu J_\mu -J_\mu\Vert =&\sup_{x\in X}\frac{|\sum_{\ell=1}^n  w_\ell(x)\big(T_\mu^{\ell}J_\mu\big)(x)-J_\mu(x)|}{v(x)}\\
    =&\sup_{x\in X}\frac{|\sum_{\ell=n+1}^\infty  w_\ell(x)J_\mu(x)|}{v(x)}.
\end{align*}
Since $\forall n$, $\exists x$ such that $x>n$. Therefore, we have $\Vert T^{(w_n)}_\mu J_\mu -J_\mu\Vert=1$ for all $n$. This implies the sequence does not converge in norm. Otherwise, its limit in norm at all $x$ would have same values as $J_\mu(x)$.  
\end{exmp}

On the other hand, the monotonicity property of $T^{(w)}_\mu$ and $T^{(\lambda)}_\mu$ follows from the monotonicity property of $T_\mu$. This is summarized in the following theorem. The proof is omitted.
\begin{lemma}\label{thm:3}
	Let the set of mappings $T_\mu:\mathcal{B}(X)\to \mathcal{B}(X)$, $\mu\in\mathcal{M}$, satisfy Assumptions~\ref{asm:contra} and \ref{asm:mon}. Then the mappings $T_\mu^{(w)}:\mathcal{B}(X)\to\mathcal{B}(X)$ defined in Eq.~\eqref{eq:thm1def} is monotonic in the sense that 
	\begin{equation*}
    J\leq J'\implies T_\mu^{(w)}J\leq T_\mu^{(w)}J',\;\forall x\in X,\,\mu\in \mathcal{M}.
\end{equation*}
\end{lemma}

\section{Convergence of $\lambda$-PIR}\label{sec:conv}
We summarize the convergence results of $\lambda$-PIR under the classical contractive model assumptions. 
\begin{theorem}\label{thm:42}
Let Assumptions~\ref{asm:contra}, \ref{asm:mon}, \ref{asm:att}, and \ref{asm:comm} hold. Given $J_0\in\mathcal{B}(X)$ such that $TJ_0\leq J_0$, the sequence $\{J_k\}_{k=0}^\infty$ generated by algorithm \eqref{eq:algm} converges in norm to $J^*$ with probability one. 
\end{theorem}

\begin{proof}
Since $TJ_0\leq J_0$, we have $T_{\mu^0}J_0=TJ_0\leq J_0$. By monotonicity of $T_{\mu^0}$, we have 
\begin{equation*}
    T_{\mu^0}^\ell J_0\leq T_{\mu^0}^{\ell-1} J_0,\;T_{\mu^0}^\ell J_0\leq T J_0,\;\forall \ell,
\end{equation*}
which implies that 
\begin{equation*}
    T_{\mu^0}^{(\lambda)}J_0\leq T_{\mu^0}J_0\leq J_0,
\end{equation*}
which means $J_1$ is upper bounded by $TJ_0$ with probability one. In addition, we also have $J_{\mu^0}\geq J^*$ where $J_{\mu^0}$ is the fixed point of both $T_{\mu^0}^{(\lambda)}$ and $T_{\mu^0}$, then we have 
\begin{equation}
    J^*\leq J_{\mu^0}\leq T_{\mu^0}^{(\lambda)}J_0\leq T_{\mu^0}J_0.
\end{equation}
which means $J_1$ is lower bounded by $J^*$ with probability one. Due to uniform-contraction Assumption~\ref{asm:contra}, we have  
\begin{align}
    &T^2 J_0 = T\big(T_{\mu^0}J_0\big)\leq T_{\mu^0}J_0,\label{eq:inprob1}\\
    &T\big(T_{\mu^0}^{(\lambda)}J_0\big)\leq T_{\mu^0}\big(T_{\mu^0}^{(\lambda)}J_0\big)\leq T_{\mu^0}^{(\lambda)}J_0,\label{eq:inprob2}
\end{align}
where in the second inequality \eqref{eq:inprob2}, we relied on the fact that $T_\mu^{(\lambda)}$ and $T_\mu$ can commute, which is due to Assumption~\ref{asm:comm}, and the fact that $T_{\mu^0}^{(\lambda)}$ is monotone, which is due to Lemma~\ref{thm:3}. Therefore, with $TJ_1\leq J_1$ with probability one. Then we can proceed by induction to show that the sequence $\{J_k\}$ is lower bounded by $J^*$ and upper bounded by sequence $\{T^kJ_0\}$ with probability one. Then due to Lemma~\ref{thm:con}, we have $\lim_{k\to\infty}\Vert J_k-J^*\Vert=0$ with probability one. 
\end{proof}

The following result, as a special case of Theorem~\ref{thm:42}, shows that if $H(\cdot,\cdot,\cdot)$ has certain `linear' structure, the initialization condition $TJ_0\leq J_0$ required in Theorem~\ref{thm:42} can be dropped and the same convergence result still stands. The proof is obtained by applying Theorem~\ref{thm:42} and the arguments in \cite{bertsekas1996temporal} and \cite[Chapter 2]{bertsekas1996neuro}.
\begin{corollary}\label{cor:1}
Let $H(\cdot,\cdot,\cdot)$ have the form
\begin{equation}
\label{eq:hlin}
    H(x,u,J)=\int_X \big( g(x,u,y)+\alpha J(y)\big)d \mathbb{P}(y|x,u)
\end{equation}
where $g:X\times U\times X\to\mathbb{R}$, $\alpha\in(0,1)$ and $\mathbb{P}(\cdot|x,u)$ is the probability measure conditioned on $(x,u)$ for certain MDP. Let $v(x)=1$ $\forall x\in X$, and Assumptions~\ref{asm:contra}, \ref{asm:mon}, \ref{asm:att}, and \ref{asm:comm} hold. Given arbitrary $J_0\in\mathcal{B}(X)$, the sequence $\{J_k\}_{k=0}^\infty$ generated by algorithm \eqref{eq:algm} converges in norm to $J^*$ with probability one. 
\end{corollary}

\begin{remark}
One key insight given in \cite{bertsekas1996temporal} is that when $H$ has `linear' form similar to \eqref{eq:hlin}, a constant shift of the cost function $J$ does not alter the choice of the optimal policy, which justifies the importance of resembling the `shape', rather than the `value', of the optimal costs in the approximation schemes. This is evidently explained in \cite[Chapter 3]{bertsekas2019reinforce}.
\end{remark}

\section{Application to ADP}\label{sec:aadp}
In this section, we exemplify the proposed algorithm for applications of ADP used to solve on-line constrained optimal control problems.

\subsection{Constrained optimal control and ADP}
Consider optimal control problems with
\begin{equation}
    x_{k+1}=f(x_k,u_k),\;H(x,u,J)=g(x,u)+\alpha J(f(x,u)),
\end{equation}
where $X\subset\mathbb{R}^n$ and $U\subset\mathbb{R}^m$ are compact sets, and $v(x)=1$ $\forall x\in X$. In addition, we assume the distribution of $x_0$, denoted as $\mathcal{X}_0$, is given. We denote collectively the problem data as $\mathbf{D}$. Assume $\mathbf{D}$ fulfills contractive model assumption, then there exists $J^*\in\mathcal{R}(X)$ such that $J^*=TJ^*$. However, it is often intractable to compute $J^*$. Instead, we aim to obtain $\Tilde{J}$, a good estimate of $J^*$. Once $\Tilde{J}$ is available, at every instance $k$, the ADP approach to control the system is to solve online a constrained optimization problem $u_k\in\argmin_{u\in U(x)}H(x_k,u,\Tilde{J})$.

The approximation of $\lambda$-PIR implementation comes from two sources. First, the estimate of $J^*$ often uses some form of parametric approximation. In this case, we consider $\Tilde{J}(x,\theta)$, where $\theta\in \Theta$ is the parameter to be trained. Second, the $T^{(\lambda)}_\mu$ operation on $\Tilde{J}$ can only be performed approximately. 

Here we exemplify an data-driven least square evaluation implementation. Our implementation follows closely the projection by Monte Carlo simulation method detailed in \cite[Section 5.5]{bertsekas2019reinforce}. Similar textbook treatment includes \cite[Chapter 5]{busoniu2017reinforcement}. Denote $\Tilde{J}(\cdot,\theta)$, $\Theta$, $\lambda$, number of training iterations $K$, and probability sequence $\{p_k\}_{k=1}^K$, collectively as $\mathbf{A}$. In addition, denote as $\text{Ber}(\cdot)$ the Bernoulli distribution and as $\text{Ge}(\cdot)$ the geometric distribution. The algorithm is summarized in Algorithm~\ref{alg:1}. At a typical training iteration $k$, the algorithm starts by sampling from $\text{Ber}(p_k)$ to decide by \eqref{eq:algm} if the cost estimate of this iteration is obtained via applying $T_{\mu^k}$ or $T_{\mu^k}^{(\lambda)}$. For every sample pair $(x_0,v)$, the state $x_0$ is drawn from $\mathcal{X}_0$, which is part of the problem data. When the $T_{\mu^k}$ step is chosen, for all $x_0$'s, their corresponding $v$'s are set to equal to $(T_{\mu^k}\Tilde{J})(x_0)$, with $\mu^k$ defined by \eqref{eq:algm}. If $T_{\mu^k}^{(\lambda)}$ is selected, an integer $\ell$ is drawn from $\text{Ge}(\lambda)$ for every $x_0$, and its corresponding $v$ is set to $(T_{\mu^k}^\ell\Tilde{J})(x_0)$. In total, it collects a size of $S$ sample pairs $(x_0,v)$, and updates the parameter $\theta$ by solving a lease square regression problem.  

\begin{algorithm2e}[ht]
\label{alg:1}
\caption{Data-driven $\lambda$-PIR}
\label{alg:lpi}
\DontPrintSemicolon
\KwIn{problem data $\mathbf{D}$, algorithm data $\mathbf{A}$, initial parameter $\theta_0$, sample size $S$}
\KwOut{$\theta$, the trained parameter}
$\theta\leftarrow \theta_0$\;
\For{$k\leftarrow 1$ \KwTo $K$}{
Initialize $\mathbf{x}\in\mathbb{R}^{n\times S}$, $\mathbf{v}\in\mathbb{R}^{S}$, $b\sim\text{Ber}(p_k)$\;
\For{$s\leftarrow 1$ \KwTo $S$}{
$x_0\sim\mathcal{X}_0$\;
\uIf{$b==1$}{$v\leftarrow \inf_{u\in U(x_0)}\big(g(x_0,u)+\alpha \Tilde{J}(f(x_0,u),\theta)\big)$}
\Else{$L\sim\text{Ge}(\lambda)$, $v=0$, $x\leftarrow x_0$\;
\For{$\ell \leftarrow 0$ \KwTo $L-1$}{
$u\in\argmin_{u'\in U(x)}\big(g(x,u')+\alpha \Tilde{J}(f(x,u'),\theta)\big)$,\\
$v\leftarrow v+\alpha^{\ell}g(x,u)$, $x\leftarrow f(x,u)$
}
$v\leftarrow v+\alpha^{L}\Tilde{J}(x,\theta)$
}
$\mathbf{x}_s=x_0$, $\mathbf{v}_s=v$
}
$\theta\in \argmin_{\theta'\in\Theta}\sum_{s\in S}|\Tilde{J}(\mathbf{x}_s,\theta')-\mathbf{v}_s|^2$
}
\end{algorithm2e}

\subsection{Numerical examples}
We apply the proposed algorithm to train the cost function used in ADP for constrained linear and nonlinear systems. Both the training and on-line ADP control problems in the examples are identified as convex and are solved by \texttt{cvxpy} (\cite{cvxpy}).
\begin{exmp}\label{ex:61}
Consider a linear scalar control problem with problem data given as:
\begin{equation*}
    x_{k+1}=x_k-0.5u_k,\,H(x,u,J)=x^2+u^2+0.95J(x-0.5u),\,\Tilde{J}(x,\theta)=ax^2+b,
\end{equation*}
where $\theta=(a,b)$, $X=[-100,100],\,U=[-1,1]$ and $\Theta=\{(a,b)\,|\,a\geq 0\}$. Similar problems have appeared in \cite{wang2015approximate,banjac2019data}. The results are shown in Fig.~\ref{fig:61}, where the performance is greatly improved from initial guess of $\theta$ after 2 iterations.  
\begin{figure}[h]
\centering     
\subfigure[System behavior under ADP controls with different cost function estimate.]{\label{fig:61a}\includegraphics[width=75mm,trim={0 0 0 5mm},clip]{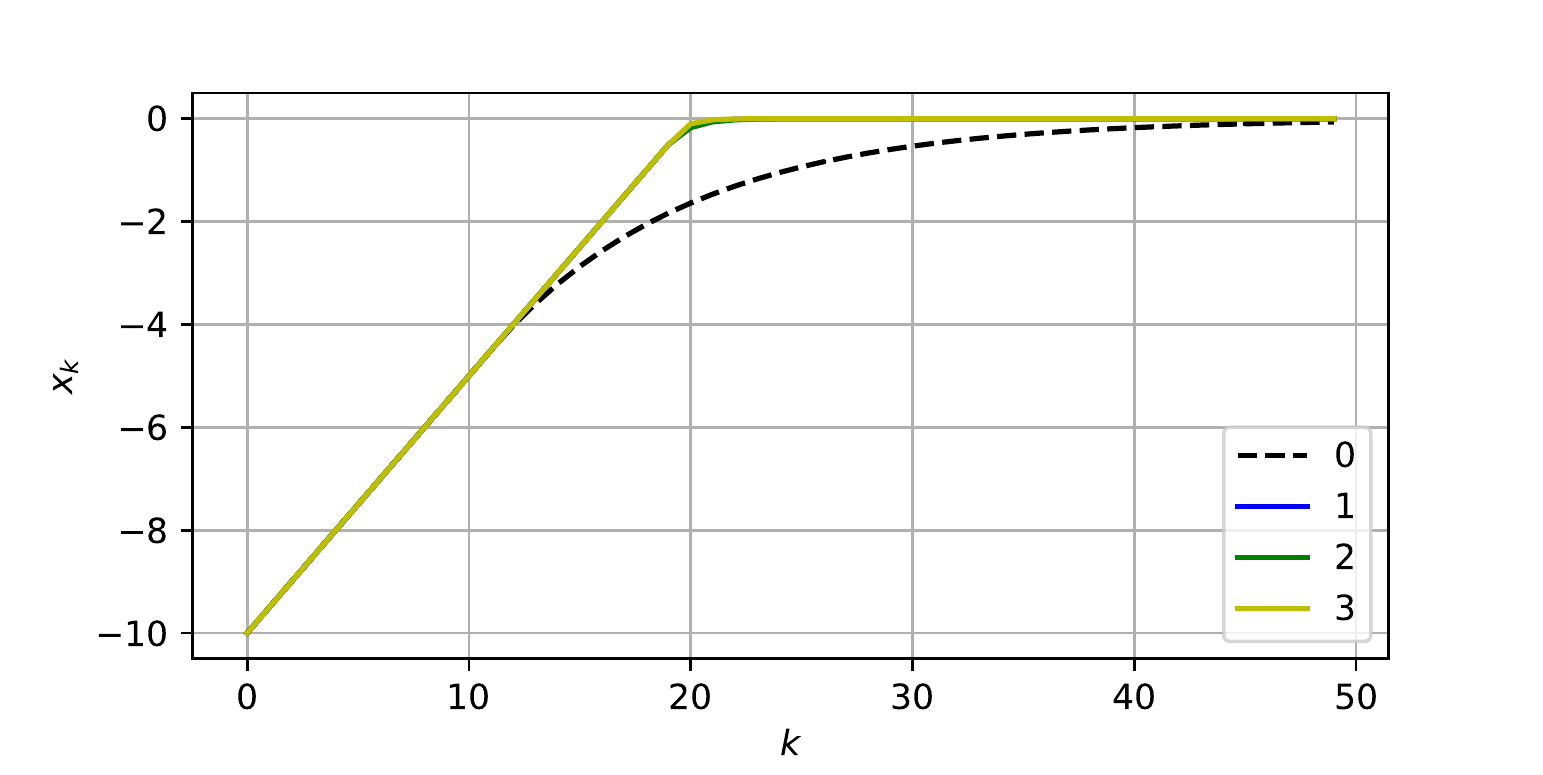}}
\subfigure[System behavior with trained and untrained cost function from different initial states.]{\label{fig:61b}\includegraphics[width=75mm,trim={0 0 0 5mm},clip]{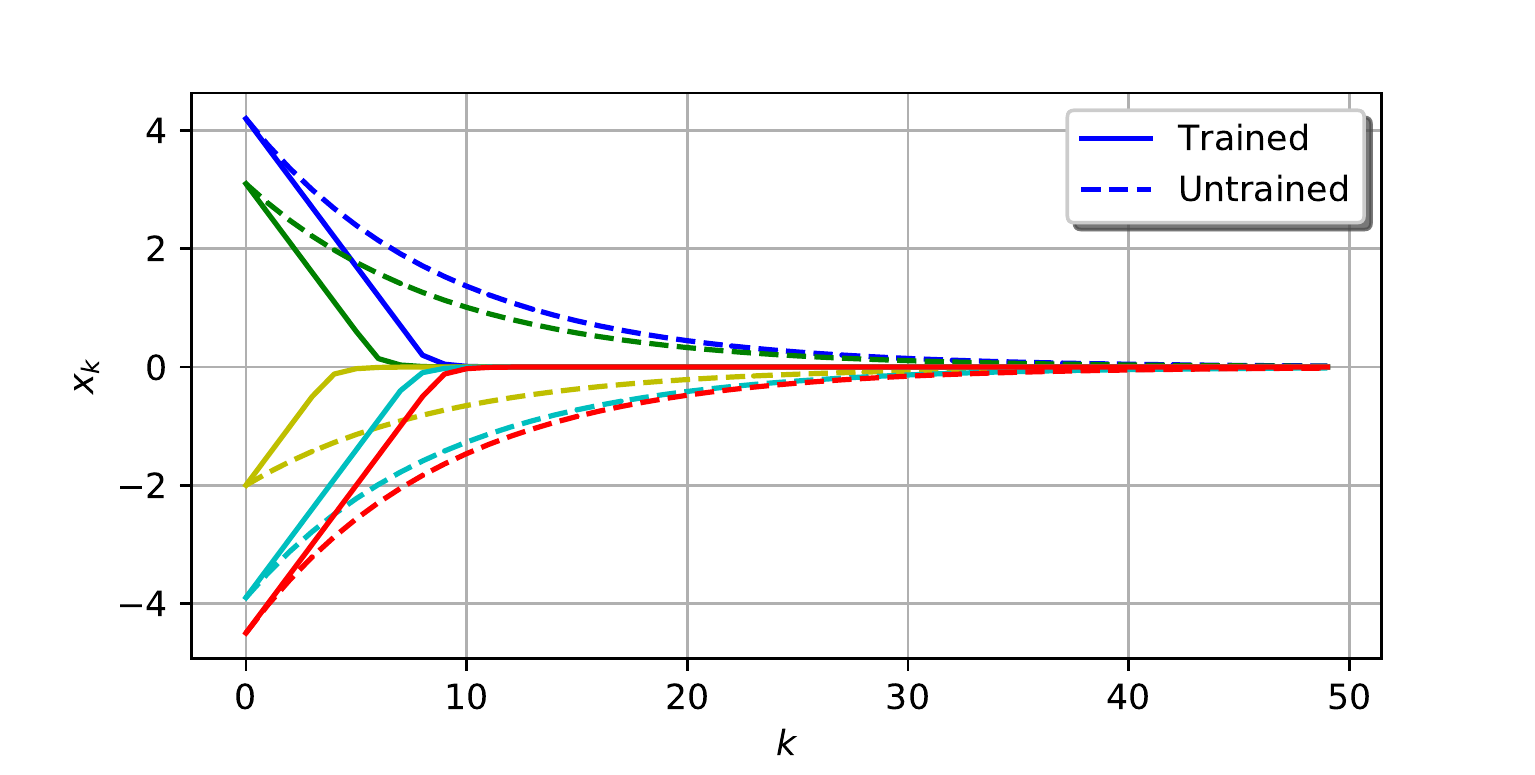}}
\caption{Closed-loop system behavior under ADP control.}
\label{fig:61}
\end{figure}

\end{exmp}

\begin{exmp}\label{ex:62}
Consider a torsional pendulum system with dynamics given as:
\begin{equation*}
    \Dot{\phi}=\omega,\,\Dot{\omega}=M^{-1}(-mgl\sin{\phi}-\gamma\omega+\tau),\,\phi\in(-\pi/2,\pi/2),\,\omega\in[-2,2],\,\tau\in[-1,1],
\end{equation*}
where $m=1/3\emph{ kg}$, $l=3/2\emph{ m}$, $M=4/3ml^2$, $\gamma=0.2$ and $g=9.8\emph{ m}/\emph{s}^2$. The discrete dynamics, denoted as $f(\cdot)$ and used for ADP control, is obtained by forward Euler method with sampling time $0.1\emph{ s}$ where the state is $x=[\phi,\omega]^\emph{T}$ with $\emph{T}$ denoting transpose operation, and the control is $u=\tau$. Then the problem data is given as:
\begin{equation*}
    x_{k+1}=f(x_k,u_k),\,H(x,u,J)=x^\emph{T}Qx+u^\emph{T}Ru+0.95J(f(x,u)),\,\Tilde{J}(x,\theta)=x^\emph{T}Px+b,
\end{equation*}
where $Q$ is identity matrix, $R=0.1$, $\theta=(P,b)$, $X\subset\mathbb{R}^2$, $U=[-1,1]$ and $\Theta=\{(P,b)\,|\,P\succeq0\}$. Similar example has appeared in \cite{si2001online,liu2013policy}. We set $S=100$, $K=5$, $p_k=0.5$ for all $k$, and $\lambda=0.1$ so that the lookahead steps in average is $1/\lambda=10$. The closed loop system behavior with initial $\theta$ and $\theta$ after $5$ iterations are shown in Fig.~\ref{fig:62} where the continuous system dynamics is solved by \emph{\texttt{ode45}}. The control performance is greatly improved.
\begin{figure}[h]
\centering     
\subfigure[With initial guess of $\theta$.]{\label{fig:62a}\includegraphics[width=75mm,trim={0 0 0 5mm},clip]{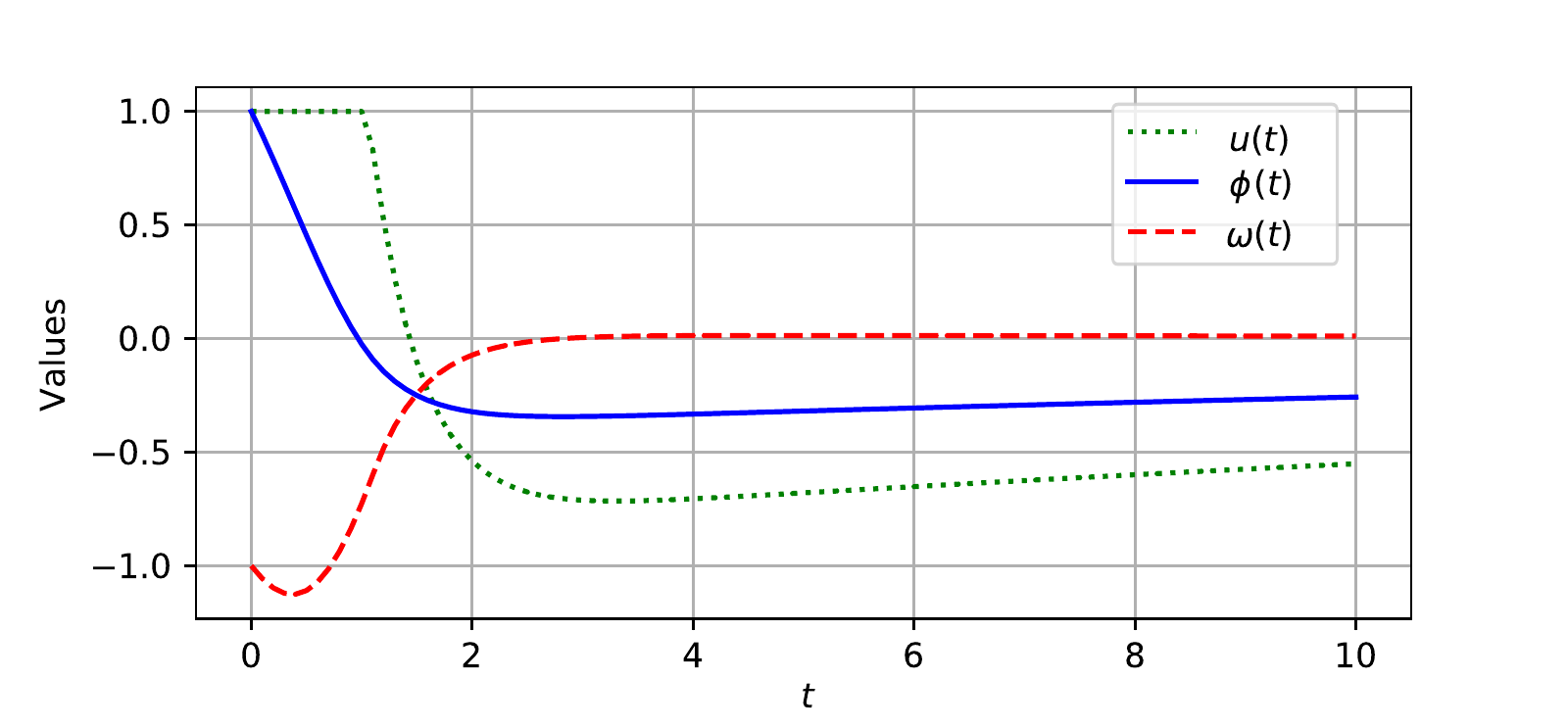}}
\subfigure[With $\theta$ after 5 training iterations.]{\label{fig:62b}\includegraphics[width=75mm,trim={0 0 0 5mm},clip]{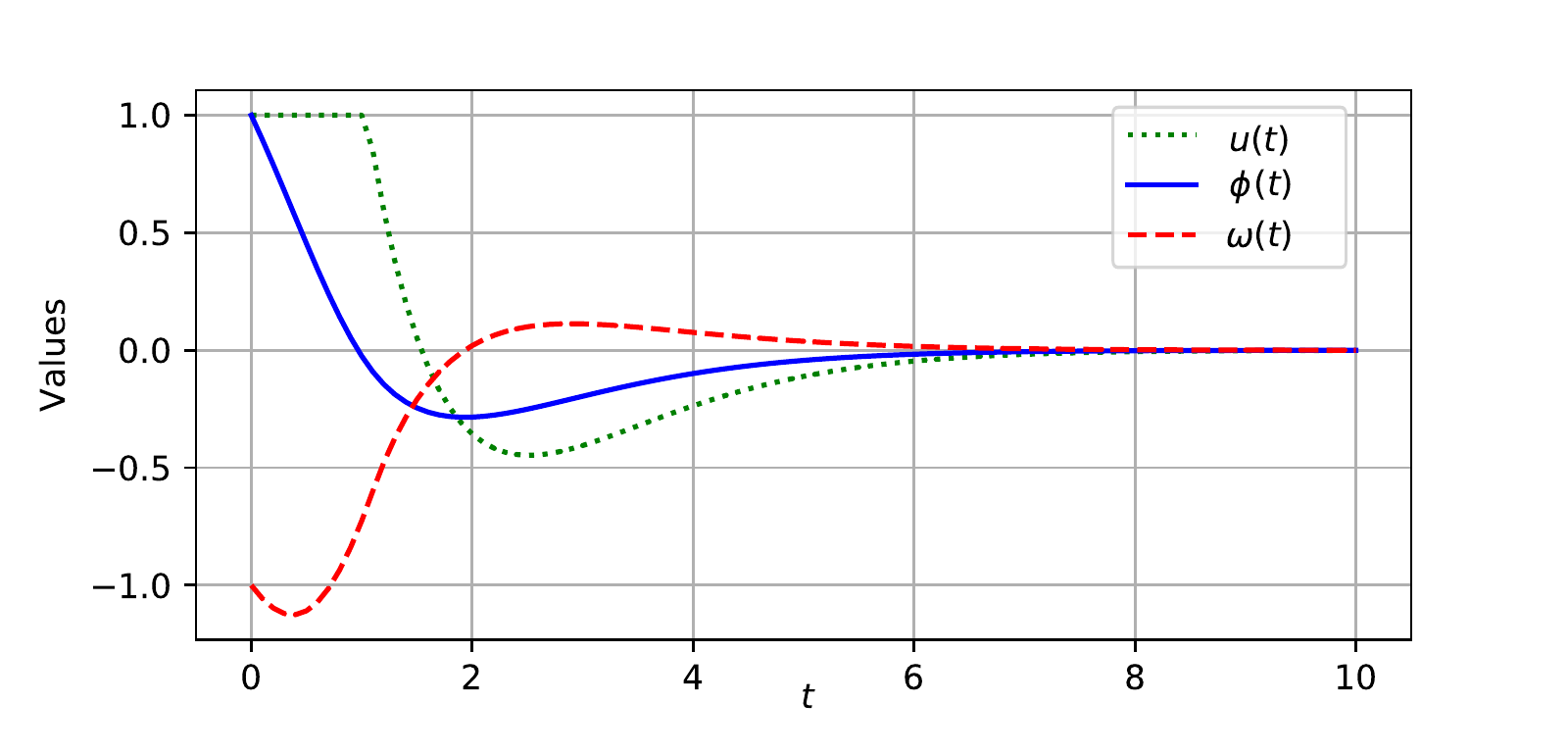}}
\caption{Closed-loop system behavior under ADP control with untrained and trained $\theta$.}
\label{fig:62}
\end{figure}

Here we show the cost function plots along the axes where $\omega=0$ and $\phi=0$, and the cost estimates converged. Besides, we compared the performance of $\lambda$-PIR with approximated implementation of VI where in \eqref{eq:algm} the evaluation step is always applying $T_{\mu^k}\Tilde{J}$, and optimistic policy iteration (OPI) where the evaluation is performed as $T_{\mu^k}^{\ell}\Tilde{J}$ with $\ell$ fixed at $1/\lambda=10$. OPI is analyzed by \cite{scherrer2015approximate} for finite state case and is known to be closed related to $\lambda$-PI. In $\lambda$-PIR, the $T_{\mu^k}^{(\lambda)}\Tilde{J}$ step occurred in the $2$nd iteration, and one can observe a `boost' towards the optimal in Fig.~\ref{fig:62ca}, while VI in Fig.~\ref{fig:62via} is yet to converge in the $5$th iteration. On the other hand, although OPI in Fig.~\ref{fig:62opib} behaves quite similarly to $\lambda$-PIR, it does require more sampling efforts compared to $\lambda$-PIR. The results here imply that $\lambda$-PIR combined the benefits of those two methods. 
\begin{figure}[h]
\centering     
\subfigure[Cost function along the axis $\omega=0$ after different training iterations.]{\label{fig:62ca}\includegraphics[width=75mm,trim={0 0 0 5mm},clip]{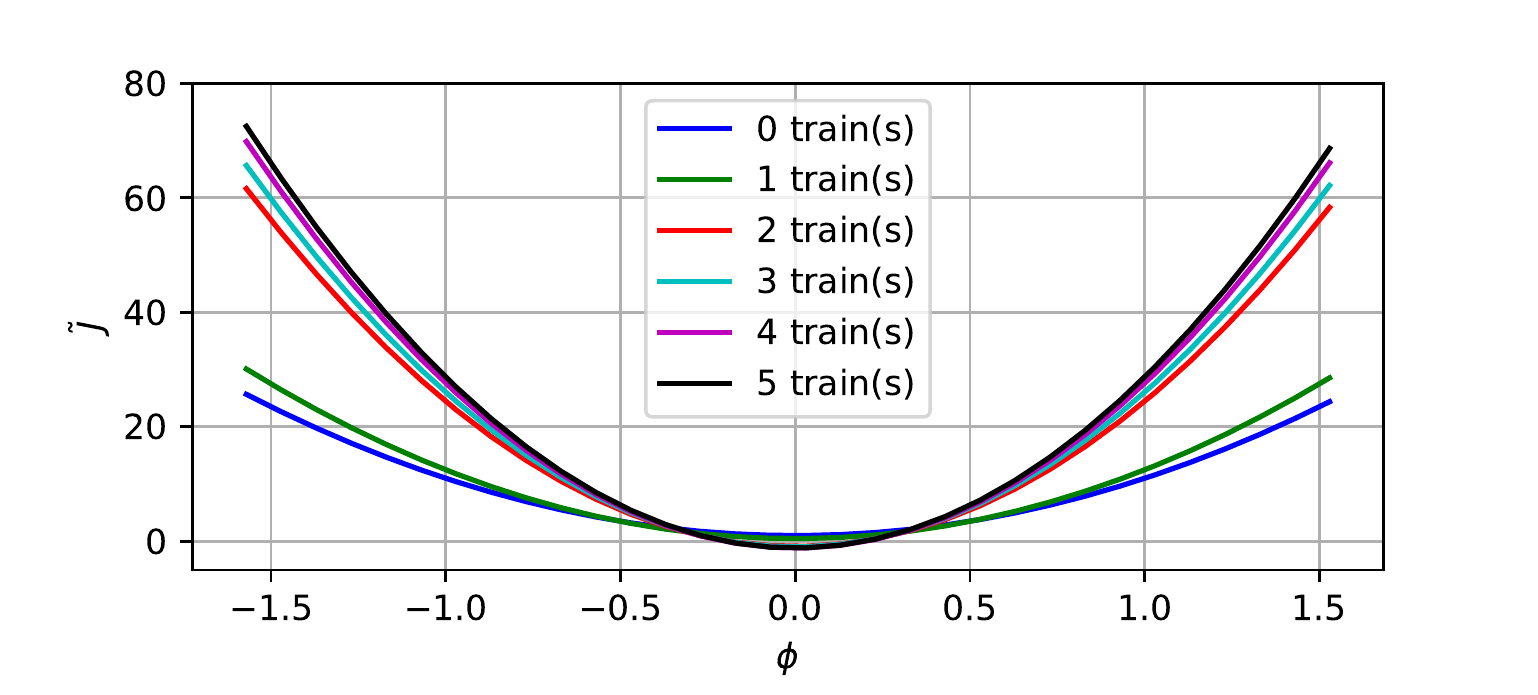}}
\subfigure[Cost function along the axis $\phi=0$ after different training iterations.]{\label{fig:62cb}\includegraphics[width=75mm,trim={0 0 0 5mm},clip]{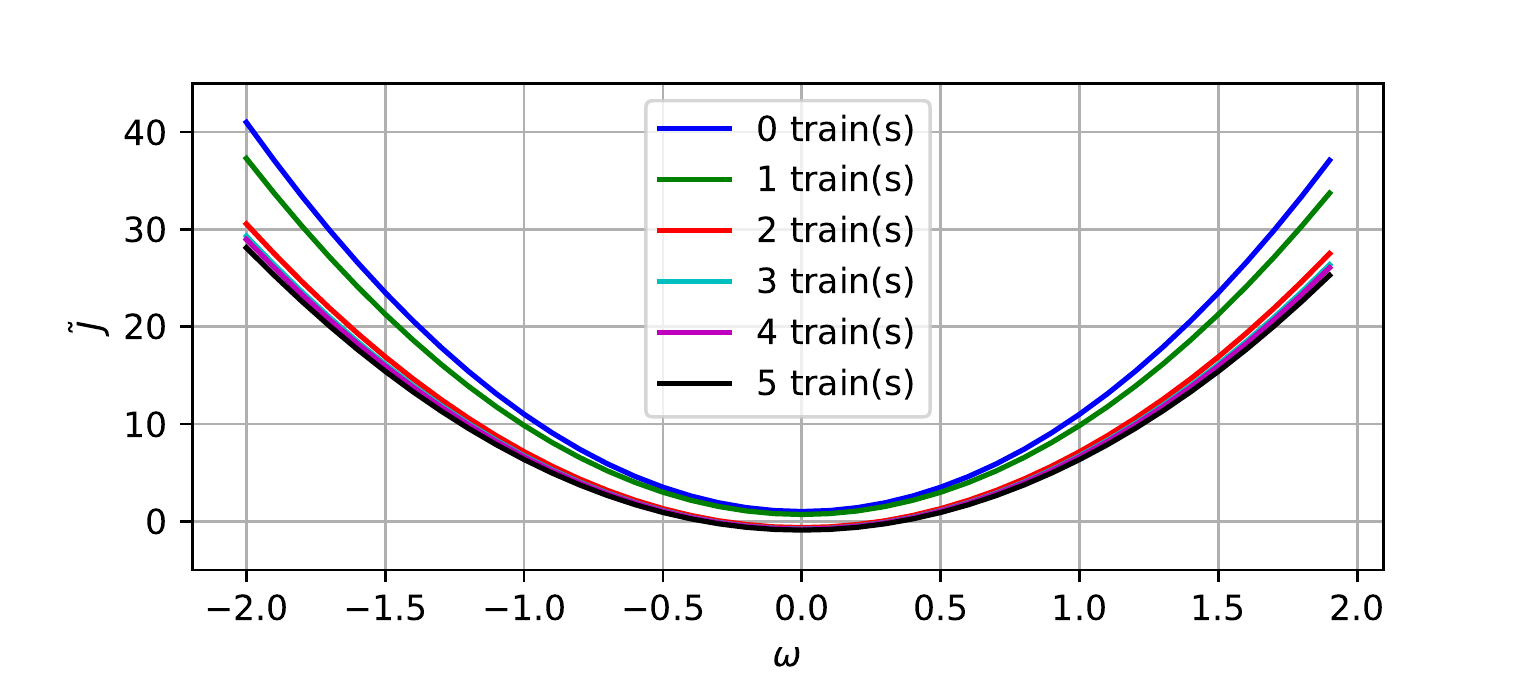}}
\caption{Cost function estimates along the axes $\phi=0$ and $\omega=0$ after different training iterations.}
\label{fig:62c}
\end{figure}

\begin{figure}[h]
\centering     
\subfigure[Cost function along the axis $\omega=0$ for VI.]{\label{fig:62via}\includegraphics[width=75mm,trim={0 0 0 3mm},clip]{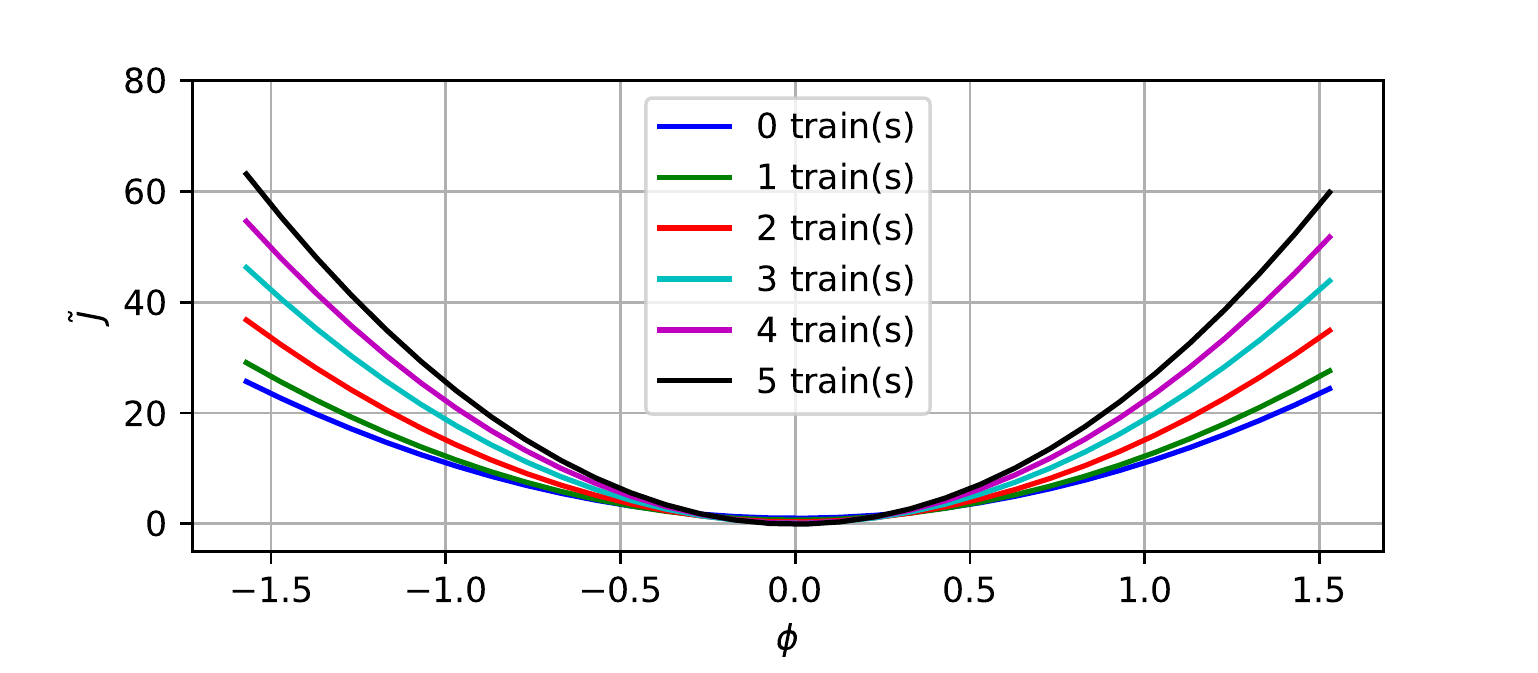}}
\subfigure[Cost function along the axis $\omega=0$ for OPI.]{\label{fig:62opib}\includegraphics[width=75mm,trim={0 0 0 3mm},clip]{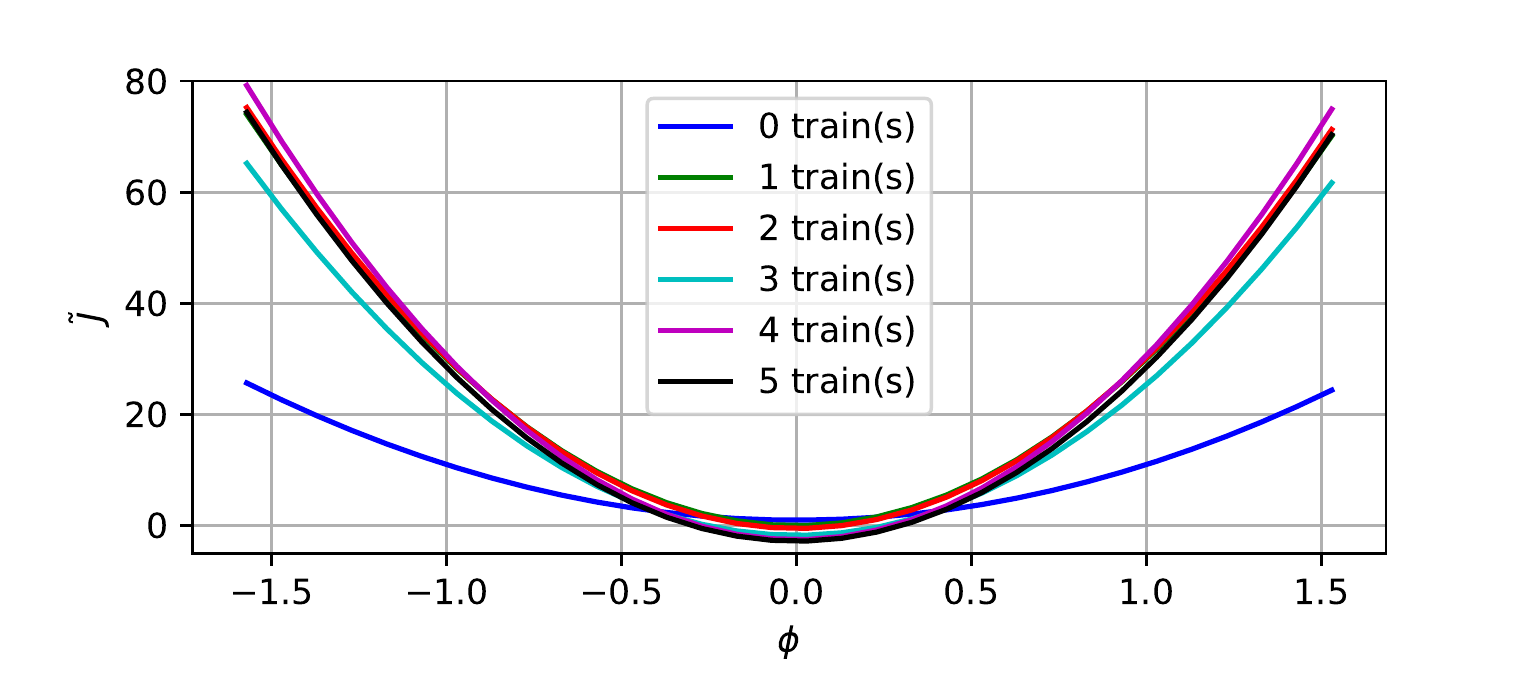}}
\caption{Cost function estimates of VI and OPI along the axis $\omega=0$.}
\label{fig:62c_vipi}
\end{figure}
\end{exmp}

\begin{exmp}
Consider a nonlinear system with dynamics given as
\begin{equation*}
    \Dot{y}=a\sin{z},\,\Dot{z}=-y^2+v,\,y\in[-2,2],\,z\in (-\pi/2,\pi/2),\,v\in[-1,1],
\end{equation*}
where $a$ is some constant. This is Example 13.13 in \cite{khalil2001nonlinear} and the goal is to set $y$ to $1$. By state feedback linearization method, we can obtain a controller given as
\begin{equation*}
    v=y^2-\frac{l_1(y-1)+l_2a\sin{z}}{a\cos{z}}
\end{equation*}
where $l_1$ and $l_2$ are parameters to be designed which impact the closed-loop poles. Denote $x=[y-1,z]^\emph{T}$ and $u=v$, and following the same procedure and using the algorithm data as in Example~\ref{ex:62}, we can obtain a cost estimate $\Tilde{J}(x,\theta)$. Fig.~\ref{fig:63} shows a comparison of system behavior under state feedback linearization and under ADP control. The constant $a$ is set to $1$ and the poles are chosen to be both at $-1$ so that the control constraint is not violated. One can see that the control response in the ADP control case is faster. The cost function estimates along axes $z=0$ and $y=0$ are shown in Fig.~\ref{fig:62c}, where the cost estimates converge after 5 iterations.
\begin{figure}[h]
\centering     
\subfigure[System behavior and control signal with state feedback linearization.]{\label{fig:63a}\includegraphics[width=75mm]{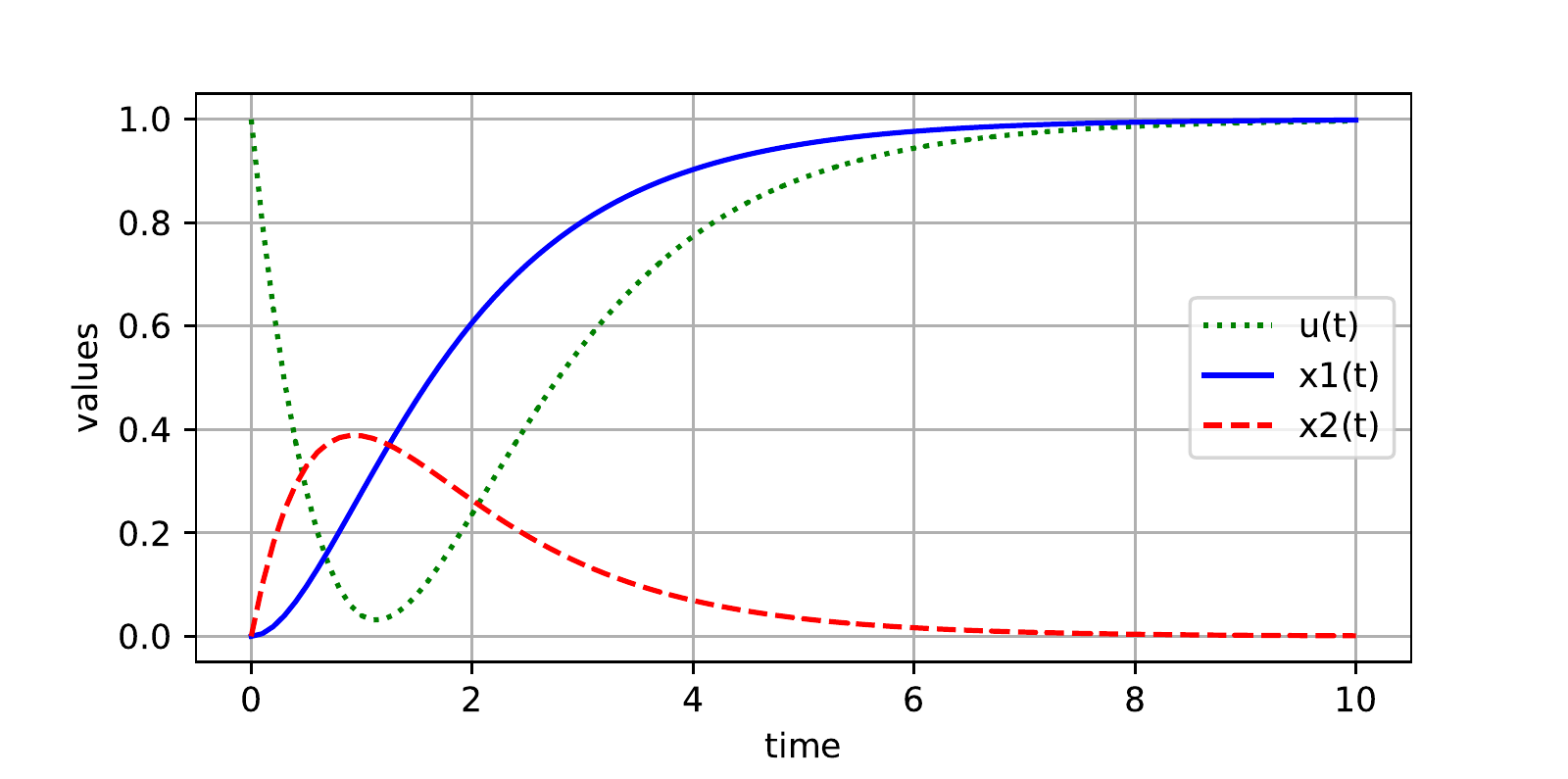}}
\subfigure[System behavior and control signal with ADP control with trained $\theta$.]{\label{fig:63b}\includegraphics[width=75mm]{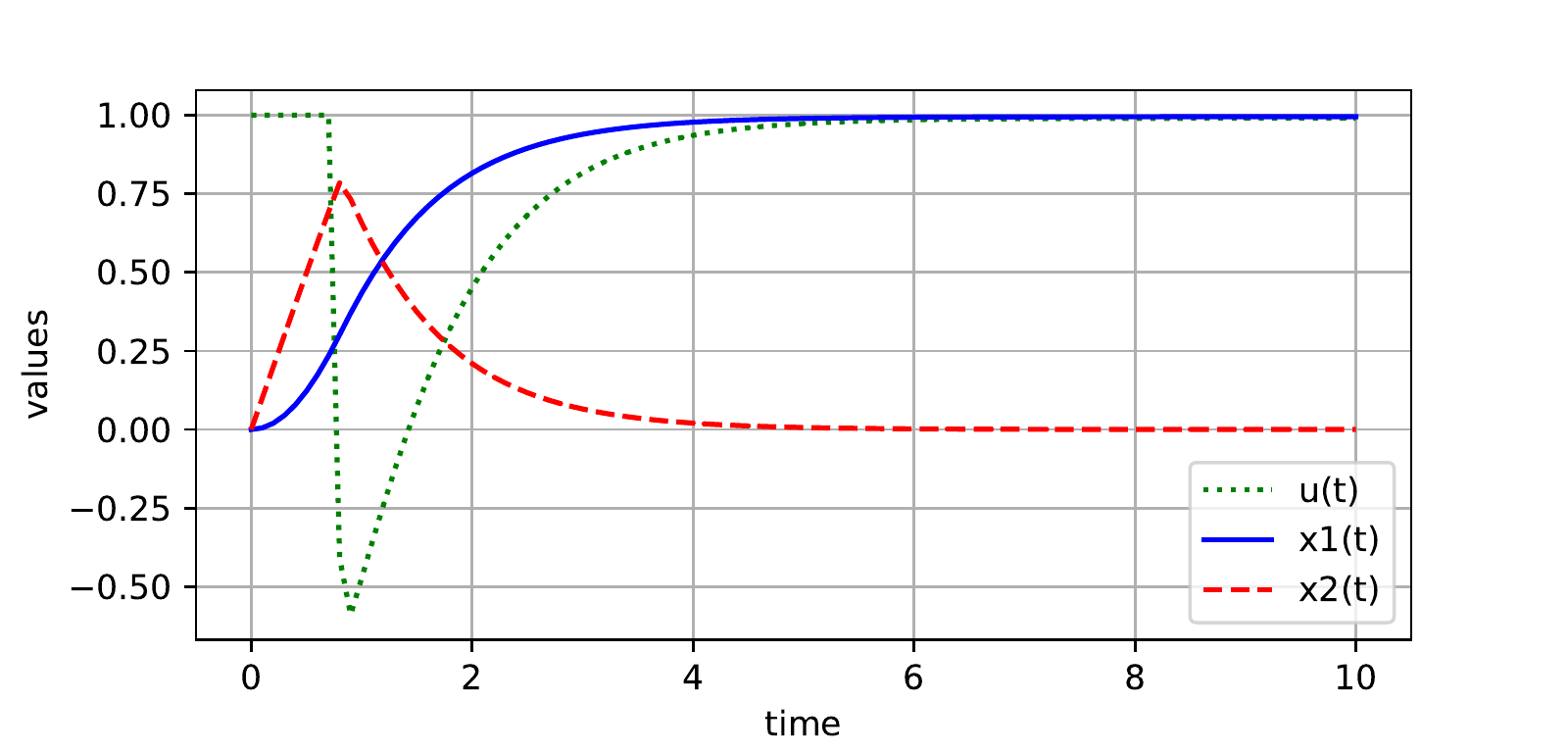}}
\caption{Closed-loop system behavior under state feedback linearization and ADP with trained $\theta$.}
\label{fig:63}
\end{figure}

\begin{figure}[h]
\centering     
\subfigure[Cost function along the axis $z=0$ after different training iterations.]{\label{fig:63ca}\includegraphics[width=75mm]{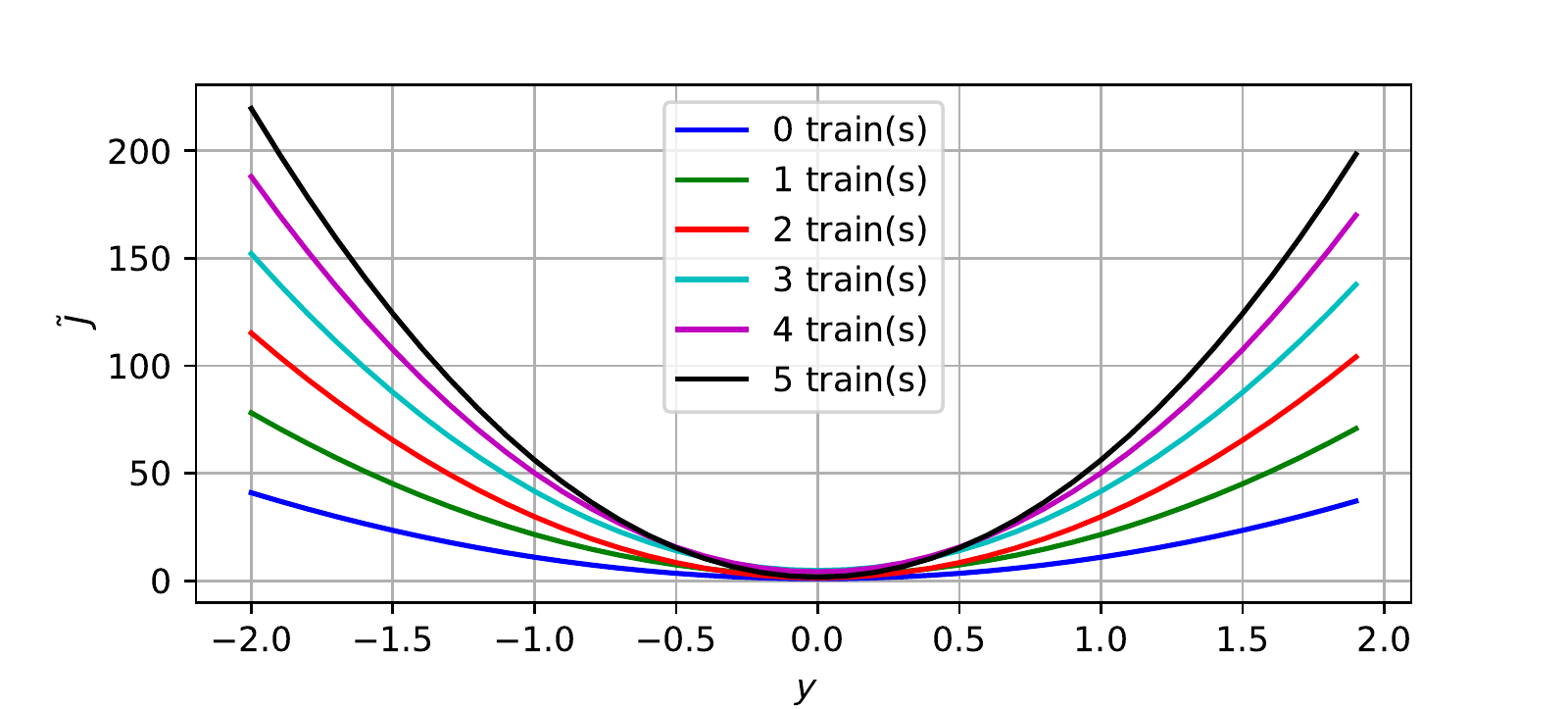}}
\subfigure[Cost function along the axis $y=0$ after different training iterations.]{\label{fig:63cb}\includegraphics[width=75mm]{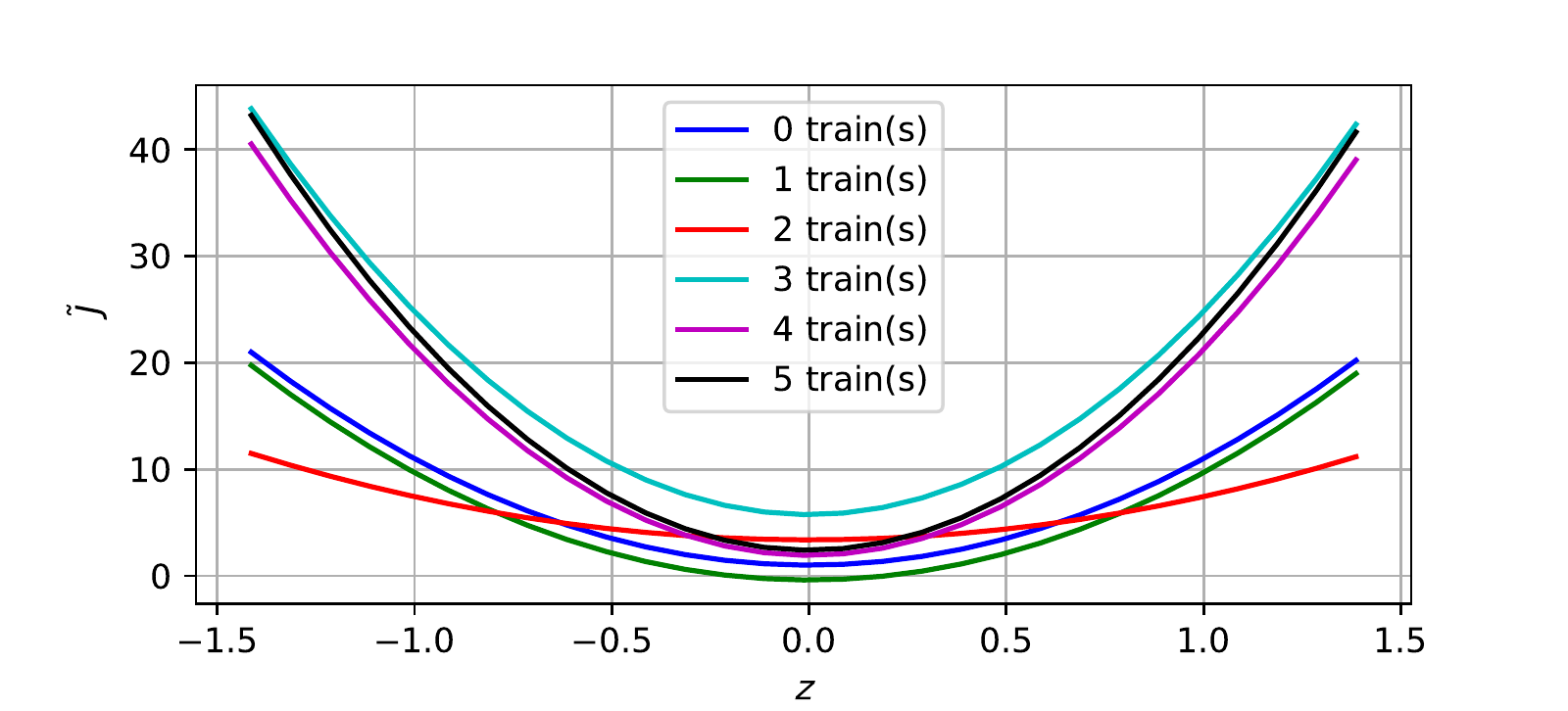}}
\caption{Cost function estimates along the axes where $z=0$ and $y=0$ after different iterations of training.}
\label{fig:63c}
\end{figure}
\end{exmp}

\section{Conclusions}\label{sec:con}
We presented results related to $\lambda$-PIR aided by abstract DP models. The $\lambda$-PIR is originally devised for finite policy problems and our results showed that the algorithm is also well-defined for contractive models with infinite states and the algorithmic convergence can be ensured for problems with infinite policies by adding an additional condition, which can be dismissed if the problem exhibits a linear structure. We exemplified a data-driven approximated implementation of the algorithm to estimate cost functions for constrained optimal control problems and the obtained estimates resulted in good closed-loop behavior when embedded in ADP for online control in numerical examples.   

\acks{This work was supported by the Swedish Foundation for Strategic Research, the Swedish Research Council, and the Knut and Alice Wallenberg Foundation. The authors are grateful to Prof.~Dimitri P. Bertsekas for discussions pointing to abstract DP models and to the connection between $\lambda$-PI and proximal algorithms, and for his suggestions to improve this work. The helpful comments from the reviewers are also acknowledged.}

\bibliography{ref}

\end{document}